\documentclass[sn-mathphys,Numbered]{sn-jnl}% Style for submissions to Nature Portfolio journals

\usepackage{graphicx}%
\usepackage{multirow}%
\usepackage{amsmath,amssymb,amsfonts}%
\usepackage{amsthm}%
\usepackage{mathrsfs}%
\usepackage[title]{appendix}%
\usepackage{xcolor}%
\usepackage{textcomp}%
\usepackage{manyfoot}%
\usepackage{booktabs}%
\usepackage{algorithm}%
\usepackage{algorithmicx}%
\usepackage{algpseudocode}%
\usepackage{listings}%

\newtheorem{theorem}{Theorem}%  meant for continuous numbers
\newtheorem{proposition}[theorem]{Proposition}% 
\newtheorem{remark}{Remark}%

\newtheorem{definition}{Definition}%
\newtheorem{lemma}[theorem]{Lemma}
\newtheorem{corollary}[theorem]{Corollary}

\newcommand{\N}{{\mathbb N}}
\newcommand{\Z}{{\mathbb Z}}
\newcommand{\K}{\mathbb K}
\newcommand{\C}{{\mathcal C}}
\newcommand{\PAut}{\rm PAut}
\newcommand{\tq}{\; \mid \;}
\newcommand{\lcm}{\mathrm{lcm}}
\newcommand{\supp}{\rm supp}

\begin{document}

\title[Permutation decoding first-order GRM codes]{Permutation decoding of first-order Generalized Reed-Muller codes.\footnote{This work was partially supported by MINECO, project PID2020- 113206GB- I00/AEI/10.13039 /501100011033, and Fundaci\'{o}n S\'{e}neca of Murcia, project 22004/PI/22.}}

\author{Jos\'e Joaqu\'{i}n Bernal}\email{josejoaquin.bernal@um.es}
\author{Juan Jacobo Sim\'on}\email{jsimon@um.es}

\affil{\orgdiv{Departamento de Matem\'{a}ticas}, \orgname{Universidad de Murcia}, \orgaddress{\country{Spain}}}
%\street{Campus de Espinardo}, \city{Murcia}, \postcode{30100},
\abstract{
In \cite{BS3} we describe a variation of the classical permutation decoding algorithm that can be applied to any binary affine-invariant code; in particular, it can be applied to first-order Reed-Muller codes successfully. In this paper we study how to implement it for the family of first-order Generalized Reed-Muller codes. Then, we give examples which show that we improve the number of errors we can correct in comparison with the known results for this family of codes. Finally, we deal, from a probabilistic point of view, with the problem of determining when the algorithm only needs to use a smaller PD-like set. }


\maketitle

\section{Introduction}

The family of (binary) Reed-Muller codes was introduced by D. E. Muller in 1954 \cite{Muller} and a specific decoding algorithm for them was presented by I. S. Reed in the same year \cite{Reed}. This family was generalized to other fields by Kasami, Lin and Peterson who introduced primitive Generalized Reed-Muller (GRM) codes in \cite{KLP}. Independently, in \cite{Weldon}, Weldon presented the nonprimitive generalization and the single-variable approach using the Mattson-Solomon polynomial. The reader may find an extensive explanation on both generalizations in \cite{DGM, handbook3} and  \cite{ handbook2}. In this paper we consider primitive GRM codes, specifically they are seen as affine-invariant codes in the algebra $\K G$, where $\K$ is the field with $q$ elements and $G$ is the abelian group of the field with $q^m$ elements, for some natural number $m$.

On the other hand, the permutation decoding was introduced by F. J. MacWilliams in \cite{macwilliams} and it is fully described  in \cite{handbook2} and \cite{MacSlo}. With a fixed information set for a given linear code, this technique uses a special set of its permutation automorphisms called $s$-PD-set, where $s$ is the number of errors we want to be corrected. Then, the idea of permutation decoding is to apply the elements of the  PD-set to the received vector until the errors are moved out of the fixed information set.

The problem of applying permutation decoding to GRM codes has been addressed for many authors earlier. Up to our knowledge, the best bounds for the number of corrected errors are in \cite{KMM2} and \cite{KMM3}. In \cite{KMM2} the authors mention that the group of translations in $\K^m$ is an $s$-PD-set for the GRM code of order $v$ where 
	$$s=\min\left\{\left\lfloor\frac{q^m-1}{f_{v,m,q}} \right\rfloor, \left\lfloor\frac{d_{v,m,q}-1}{2}\right\rfloor\right\}$$
	and $f_{v,m,q}, d_{v,m,q}$ represent the dimension and the minimum distance of the code respectively. In \cite{KMM3} the authors explain that for the first-order GRM codes, $s$-PD-sets of size $s+1$ exist for 
	$$1\leq s\leq \left\lfloor\frac{q^m}{m+1}\right\rfloor-1.$$
	
In \cite{BS} we present a description of an information set for first and second-order Reed-Muller codes only in terms of its essential parameters. Later, in \cite{BS3} we introduce a modification of the classical permutation decoding algorithm which is valid for any affine-invariant code and we show that it can be apply successfully in the case of first order Reed-Muller codes with respect to the mentioned information sets. In this paper we see that the construction of the information sets can be generalized to first-order GRM codes and then we show how to apply the modified permutation decoding algorithm to that  family of codes in order to improve the bounds cited above.

In Section \ref{Preliminaries} we include the basic notation and we establish the general ambient space for our work. In Section \ref{infosets} we show that the construction given in \cite{BS} can be generalized to first-order GRM codes and that we can take a suitable decomposition of $q^m-1$ for our purposes. Then, in Section \ref{NewPD} we recall the modified permutation decoding algorithm presented in \cite{BS3} and we set the main result for first-order GRM codes. In Section \ref{Examples} we present the results derived from the application of the permutation decoding algorithm for some values of $q$ and $m$. Moreover, we show that they improve the mentioned bounds. Finally, in the last section, we study, from a probabilistic point of view, when we can apply efficiently a smaller PD-like set. Concretely, we include some values which show that in some cases we can expect to do it with probability close to 1.

\section{Preliminaries}\label{Preliminaries}

  We are considering the family of Reed-Muller codes as a subfamily of the so-called affine-invariant codes (see, for instance, \cite{handbook2},\cite{handbook3},\cite{Charpin}). 
	
	Let $\K$ be a field with $q$ elements, $q$ a power of a prime number $p$. Our ambient space will be the group algebra $\K G$, where $G$ is the additive group of the field with $q^m$ elements, for some $m\in \N$. From now on we fixed $m$ and we define $n=q^m-1$. Observe that $G$ is an elementary abelian group of order $|G|=q^m$ and $G^*=G\setminus\{0\}$ is a cyclic group. We are also fixing $\alpha\in G^*$ a generator element, that is, $\langle \alpha \rangle=G^*$. 
	Then, we write the elements in $\K G$ as
	\begin{equation}\label{polynomial}
	 b X^0+\sum\limits_{i=0}^{n-1} a_i X^{\alpha^i}\quad (a_i,b\in \K).
	\end{equation}

\begin{definition}
\begin{enumerate}[a)]
	\item A code in $\K G$ is any ideal $\C\leq \K G$.
	\item A code $\C\leq \K G$ is an extended cyclic code if for any $b X^0+\sum\limits_{i=0}^{n-1} a_i X^{\alpha^i}\in \C$ one has that $b X^0+\sum\limits_{i=0}^{n-1} a_i X^{\alpha^{(i+1)}}\in \C$ and $b +\sum\limits_{i=0}^{n-1} a_i =0\in \K$.
\end{enumerate} 
\end{definition}

For any extended cyclic code $\C\leq \K G$ we denote by $\C^*$ the punctured code at the position $X^0$. Then, $\C^*$ is a \textit{cyclic} code in the sense that	it is the projection to $\K G^*$ of the image of a cyclic code via the map 
 \begin{eqnarray}\label{inmersionciclicos}
  \nonumber \K[X]/\langle X^n-1\rangle &\longrightarrow& \K G\\
   \sum\limits_{i=0}^{n-1}a_{i}X^i&\hookrightarrow& \left(-\sum\limits_{i=0}^{n-1} a_i\right)X^0+\sum\limits_{i=0}^{n-1}a_i X^{\alpha^i},	  
  \end{eqnarray}
	
where $\alpha$ is the fixed $n$-th root of unity and $\K[X]/\langle X^n-1\rangle$ is the quotient algebra of $\K[X]$, the polynomials with coefficients in $\K$.  
	
\begin{definition}	Let $S_G$ denote the group of automorphisms of $G$. Then, we see $S_G$ acting on $\K G$ via 
		$$\tau\left(b X^0+\sum\limits_{i=0}^{n-1} a_i X^{\alpha^i}\right)=b X^{\tau(0)}+\sum\limits_{i=0}^{n-1} a_i X^{\tau(\alpha^i)}$$
	for any $\tau \in S_G$. In this context, we define 
	$$\PAut(\C)=\{\tau\in S_{G}\mid \tau(\C)=\C\}.$$
\end{definition}

	It is relevant to note that we may also consider $S_{G^*}$, the group of automorphisms of $G^*$, and then we may define $\PAut(\C^*)$. Moreover, we can identify any $\tau\in S_{G^*}$ with the corresponding automorphism in $S_G$ that fixes the position $X^0$ and we may write $\PAut(\C^*)\subseteq \PAut(\C)$.

	Now, we may introduce the family of affine-invariant codes.	

\begin{definition}
	 We say that a code $\C\leq \K G$ is an affine-invariant code if it is an extended cyclic code and ${\rm GA}(G,G)\subseteq \PAut(\C)$ where
	$${\rm GA}(G,G)=\{x\mapsto ax+b\mid x\in\K G, a\in G^*, b\in G\}.$$
	\end{definition}
	
We are defining the family of Generalized Reed-Muller (GRM) codes in terms of its defining set as affine-invariant codes. So, let us give the definitions of defining set and $q$-weight.
\begin{definition}  Let $\C\leq\K G$ be an affine-invariant code. For any $s\in\{0,\dots,n=q^m-1\}$ we consider the $\K$-linear map $\phi_s:\K G\rightarrow G$ given by
\begin{equation*}
 \phi_s\left(b X^0+\sum\limits_{i=0}^{n-1} a_i X^{\alpha^i}\right)=0^s+\sum\limits_{i=0}^{n-1} a_i \alpha^{is}
\end{equation*}
where we assume $0^0=1\in\K$ by convention. Then the set 
 $$D(\C)=\{i\mid \phi_i(x)=0 \text{ for all } x\in \C\}$$
 is called the defining set of $\C$.
\end{definition}

It is well known that any affine-invariant code is totally determined by its defining set (see, for instance,\cite{handbook2}).

\begin{definition}
For any natural number $k$ its $q$-ary expansion is the sum 
$$\sum_{r\geq 0} k_r q^r=k$$
 with $k_r\in \{0,1,\dots,q-1\}$. The $q$-weight of $k$ is ${\rm wt}_q(k)=\sum_{r\geq 0}k_r$. 
\end{definition}

\begin{definition}\label{defRM}
 Let $0< \rho\leq m(q-1)$. The Generalized Reed-Muller (GRM) code of order $\rho$ and length $q^m$, denoted by $R_q(\rho,m)$, is the affine-invariant code in $\K G$ with defining set 
 $$D(R_q(\rho,m))=\{0\leq i<q^m-1 \mid {\rm wt}_q(i)< m(q-1)-\rho\}.$$
\end{definition}

In this paper we deal with first-order GRM-codes, that is, the codes with defining set
$$D(R_q(1,m))=\{0\leq i<q^m-1 \mid {\rm wt_q}(i)< m(q-1)-1\}.$$
%We denote by $R^*_q(1,m)$ the cyclic punctured code at the position $X^0$.

The following proposition shows some well-known results about first-order GRM-codes (see, for instance, \cite{handbook3}).

\begin{proposition}For any value of $q$ and $m$ one has that
\begin{enumerate}
	\item Dimension: ${\rm dim}_\K(R_q(1,m))=m+1$.
	\item Minimum distance: $d(R_q(1,m))=q^{m-1}(q-1)$.
	\item $R_q(1,m)^\bot=R_q(m(q-1)-2,m)$, where the dual code $R_q(1,m)^\bot$ is defined as usual.
\end{enumerate}
\end{proposition}

We denote by $R^*_q(1,m)$ the cyclic punctured code at the position $X^0$ of $R_q(1,m)$.

\section{Information sets from defining sets for first-order Generalized Reed-Muller codes}\label{infosets}

In \cite{BS} we showed how to obtain an information set for (binary) first-order Reed-Muller codes only in terms of $m$ and a decomposition of $n=q^m-1$. In this section we fix the general notation and we include the main result.

It is clear that $\gcd(q,n)=1$ so, as in \cite{BS}, we only need to impose the following restrictions:
\begin{equation}\label{condFirstOrder}
 n=r_1\cdot r_2,\qquad \gcd(r_1,r_2)=1 \qquad \text{ and } \qquad  r_1, r_2>1
\end{equation}
Then, from now on, we fix an arbitrary isomorphism $\varphi:\Z_n\longrightarrow \Z_{r_1}\times\Z_{r_2}$.

Let us give the definition of information set in the ambient space $\K G$.

\begin{definition}
 A set $I\subseteq \{0,\alpha^0,\dots,\alpha^{n-1}\}=G$ is an information set for a code $\C\leq \K G$, with dimension $k$, if $|I|=k$ and the set of the projections of the codewords of $\C$ onto the positions in $I$ is equal to $\K^{|I|}$.
\end{definition}
It is important to note that for any affine-invariant code we may also define the notion of information set for $\C^*$ as a set contained in $G^*=\{\alpha^0,\dots, \alpha^{n-1}\}$. Clearly, an information set for $\C^*$ is an information set for $\C$, but if $I\subseteq G$ is an information set for $\C$ and $0\in I$ then $I':=I\setminus\{0\}$ is not an information set for $\C^*$. Indeed, the proof of the following main result, analogous to that given in \cite{BS}, is based on constructing an information set for $R^*_q(1,m)$ that contains $0$, so $I'=I\setminus\{0\}$ is not an information set for $R^*_q(1,m)$.

\begin{theorem}[see \cite{BS}, Theorem 34] \label{teoremainfosetprimerorden}
  Let $q$ and $m$ satisfying conditions (\ref{condFirstOrder}). Let $n=q^m-1=r_1\cdot r_2$ and denote by $a$ the order of $q$ modulo $r_1$ ($a= Ord_{r_1}(q)$). Then the set $I=\{0,\alpha^i\mid i\in \varphi^{-1}\left(\Gamma\right)\}$ where 
 $$\Gamma=\Gamma(\C)=\left\{(i_1,i_2)\in\Z_{r_1}\times\Z_{r_2} \tq 0\leq i_1< a, 0\leq i_2<\frac{m}{a}\right\},$$
 is an information set for $R_q(1,m)$. 
\end{theorem} 

The proof depends on the identification of the punctured cyclic code as a ``two dimensional abelian code'' and the computation of the parameters that define $\Gamma(\C)$ (see \cite{BS3}) only in terms of the parameters $m$ and $a$. Those arguments are independent of the value $q$ so they work as in \cite{BS}. 
\medskip

\begin{remark}
In \cite{BS} we include a table showing some values of $m$ for which it is possible to find at least one suitable decomposition of $n=2^m-1$. Although there exist such decompositions for most of the values of $m$, it is not always true, for instance in the case of Mersenne prime numbers.  

For $q=3,4,5,7,8$ we have checked computationally that we will always be able to find a suitable decompositions of $n=q^m-1$ where $m\in\{2,\dots,100\}$, with the unique exception of $q=3, m=2$.

For us, it is an open problem to determine a general condition on the values of $q$ and $m$ to get a suitable decomposition.
\end{remark}

We can establish the following two technical results that yield a simplification of the information set $\Gamma(\C)$ which is essential to apply the permutation decoding algorithm. The first one is easy to prove.

\begin{lemma}\label{descomposicion1}
 Let $m,\delta\in \N$ and $q$ a power of a prime number. Then $m|\delta$ if and only if $q^m-1|q^\delta-1$.  
\end{lemma}
%\begin{proof}
  %Firstly we assume that $m|\delta$. Then, if we use the well-known formula
	%$$X^k-1=(X-1)\cdot \sum\limits_{i=1}^{k}X^{i-1} (\text{for any } k\in\N)$$
	%with $X=q^m$ and $k=\delta/m$ we have 
	%$$(q^m)^{\delta/m}-1=q^\delta-1=(q^m-1)\cdot \sum\limits_{i=1}^{\delta/m} (q^m)^{i-1}$$
	%that implies $q^m-1|q^\delta-1$.
	%
	%Now, assume that $q^m-1|q^\delta-1$. Then, $q^\delta\equiv 1 (\text{mod } q^m-1)$. On the other hand, it is clear that $Ord_{q^m-1}(q)=m$ so $m|\delta$.
%\end{proof}

\begin{lemma}\label{descomposicion2}
 Let  $m,\delta\in \N$ and $q$ a power of a prime number. Suppose that $q^m-1=r_1\cdot r_2$ with $\gcd(r_1,r_2)=1, r_1,r_2>1$. Then either $Ord_{r_1}(q)=m$ or $Ord_{r_2}(q)=m$.
\end{lemma}

\begin{proof}
 Let us denote $Ord_{r_1}(q)=a, Ord_{r_2}(q)=b$ and $\mu=\lcm(a,b), d=\gcd(a,b)$.

 Since $q^m\equiv 1 (\text{mod } r_i)$, for $i=1,2$, we have that $a,b|m$ and then $\mu|m$. On the other hand, $q^\mu\equiv 1 (\text{mod } r_i), i=1,2,$ which implies $r_1,r_2| q^\mu-1$, and so $q^m-1=r_1\cdot r_2|q^\mu-1$ because $r_1,r_2$ are coprime. By the previous lemma we have that $m\leq \mu$, and so $m=\mu$.

Now, assume w.l.o.g. that $a\leq b$. In case $a=b$ we are done because this implies $a=b=\mu=m$. So, let us suppose that $a<b$; we are proving that $b=m$.

Since $r_1|q^a-1$ and $r_2|q^b-1$ then $r_1\cdot r_2|(q^a-1)\cdot (q^b-1)$ and so
$$q^m-1=r_1\cdot r_2\leq (q^a-1)\cdot (q^b-1)=q^{a+b}+1-(q^a+q^b)<q^{a+b}-1,$$
note $q^a+q^b>2$ for any value of $q$ because $a,b>0$. This implies that $m<a+b$.

Finally,
$$\left(\frac{a}{d}\right)b<a+b\rightarrow \left(\frac{a}{d}-1\right)b<a$$
where we have used that $\mu\cdot d=a\cdot b$. Since $a<b$, we conclude that $a=d$ and therefore $b=m$.
\end{proof}

The following corollary gives the information set that we are interested in.

\begin{corollary}\label{informationsetFisrtGRM}
 Let $q$ and $m$ satisfying conditions (\ref{condFirstOrder}). Then the set $I=\{0,\alpha^i\mid i\in \varphi^{-1}\left(\Gamma\right)\}$ where 
 \begin{equation}\label{infosetfirstorder}
\Gamma=\Gamma(\C)=\left\{(i_1,i_2)\in\Z_{r_1}\times\Z_{r_2} \tq 0\leq i_1< m, 0\leq i_2<1\right\},
\end{equation}
 is an information set for $R_q(1,m)$. 

\end{corollary} 

\begin{proof}
It follows from Lemma \ref{descomposicion2} and Theorem \ref{teoremainfosetprimerorden}.
\end{proof}

\section{Permutation decoding for first-order GRM codes}\label{NewPD}

  \subsection{Permutation decoding for affine invariant codes.}
	
	We are going to use the modified permutation decoding algorithm that we presented in \cite{BS3}. It is valid for any affine-invariant code and, in particular, for the family we are interested in: first-order GRM codes. (Note that the cardinality of the base field it is not relevant in the algorithm.)  
	
	We consider an affine-invariant code $\C\leq \K G$, with error correction capability $t$, and we denote by $\C^*$ its cyclic punctured code at the position $X^0$.

\begin{definition}
 Let $J\subseteq G^*=\{\alpha^i\mid i=0,\dots,n-1\}$ be an arbitrary set. % with $\left| J\right|\leq s, s\in \N$.
 A set $P\subseteq \PAut(\C^*)\subseteq S_{G^*}$ is an $s$-PD-like set for  $\C$ with respect to $J$ if any $s$ positions in $G^*$ are moved out of $J$ by at least one element in $P$. In case $s=t$ we say that $P$ is a PD-like set. 
\end{definition}

As in the classical permutation decoding algorithm we need the following well-known result.

\begin{theorem}(\cite{MacSlo})\label{PD}
	 Let $H$ be a parity check matrix for $\C$ in standard form with respect to an information set $I$. Then, the information symbols of $r$ are correct if and only if $\omega(H\cdot r^T)\leq t$, the error correction capability of $\C$.
	\end{theorem}

The following automorphisms of $G$ play an essential role in the algorithm. By definition of affine-invariant code it is clear that they belong to $\PAut(\C)$.

\begin{definition}
  For any $k=0,1,\dots,n-1$ let $\sigma_k$ the automorphism of $G$ given by 	
	$$\sigma_k\left(b X^0+\sum\limits_{i=0}^{n-1}a_i X^{\alpha^i}\right)=b X^{\alpha^k}+\sum\limits_{i=0}^{n-1}a_i X^{(\alpha^i+\alpha^k)} $$

We denote by $\Sigma$ the sequence $(1_G,\sigma_0,...\sigma_{n-1})$, where $1_G$ represents the identity automorphism in $\K G$. 
\end{definition}
\bigskip

Now, let $r=c+e$ be the received word with $c\in \C$ and $e$ the error vector, where we assume $\omega(e)\leq t$. We fix an information set for $\C$, $I\subseteq\{0,\alpha,\dots,\alpha^{n-1}\}$, and we denote as above
\begin{equation}\label{Iprima}
I'=\left\{
\begin{array}{lcl}
 I & \text{  if  } & 0\notin I\\
 I\setminus\{0\} &  & \text{otherwise}
\end{array}
\right.
\end{equation} 

Then the (modified) permutation decoding algorithm is:

{\scshape Algorithm I:}
	
	\begin{enumerate}
		\item Let H be a parity check matrix in
standard form with respect to I, and $P\subseteq \PAut(\C^*)\subseteq\PAut(\C)$ an $s$-PD-like set for $\C^*$ ($s\leq t$) with respect to $I'$, defined in (\ref{Iprima}).

	\item We take $r=1_G(r)$. For each $\tau\in P$ we compute $\tau(r)$. If we find $\tau_0\in P$ such that $\omega\left(H\cdot \tau_0(r)^T\right)\leq t$ then the information symbols of $\tau_0(r)$ are correct (Theorem \ref{PD}).
	
	In case there isn't any $\tau_0\in P$ satisfying the desired condition we go to Step 3, otherwise we go to Step 4.
	
	\item We compute $\sigma_i(r)$, where $\sigma_i$ is the next element in $\Sigma$, and we repeat Step 2 starting with $\sigma_i(r)$ instead of $r$.
	
		\item We recover $c'\in \C$ from the information symbols of $\tau_0(r)$.
		\item We decode to $(\sigma_i\circ\tau_0)^{-1}(c')=c$.
	
\end{enumerate}

\bigskip
The following result guarantees the functionality of the algorithm:

\begin{theorem}[Theorem 13,\cite{BS3}]\label{PDnew}
 Let $\C\subseteq \K G$ be an affine-invariant code with correction capability $t$. Let $I\subseteq \{0,\alpha^0,\dots,\alpha^{n-1}\}$ be an information set for $\C$. Let $P\subseteq \PAut(\C^*)\subseteq\PAut(\C)$ be an $s$-PD-like set for $\C^*$ with respect to $I'$ where $s\leq t$. Then, we can correct up to $s$ errors by using the previous algorithm.
\end{theorem}

Therefore, to achieve our purposes we need to find $s$-PD-like sets for first-order GRM codes for a fixed information set.

\subsection{Permutation decoding for first-order GRM codes.}\label{PDRM}
Let $\C=R_q(1,m)$, for any $m\in\N$, $q$ a power of a prime number, and suppose that $n=q^m-1$ satisfies the required conditions, that is,
$$ n=r_1\cdot r_2,\qquad \gcd(r_1,r_2)=1 \qquad \text{ and } \qquad  r_1, r_2>1.
$$

Recall that $\alpha\in G^*$, an $n$-th primitive root of unity, and $\varphi$, an isomorphism $\varphi: \Z_n\longrightarrow \Z_{r_1}\times\Z_{r_2}$, have been fixed all throughout this paper.

By Corollary \ref{informationsetFisrtGRM} we can consider the information set for $\C$ given by
$$I=\{0,\alpha^i\mid i\in \varphi^{-1}(\Gamma)\}$$
where  $$\Gamma=\Gamma(\C)=\left\{(i_1,i_2)\in\Z_{r_1}\times\Z_{r_2} \tq 0\leq i_1< m, 0\leq i_2<1\right\}.$$

Now, let $T_\alpha\in S_{G^*}$ be defined by 
 \begin{equation}\label{talfa}
T_\alpha\left(b X^0+\sum_{i=0}^{n-1} a_i X^{\alpha^i}\right)=b X^0+\sum_{i=0}^{n-1} a_i X^{\alpha^{i+1}}\end{equation}
Since $\C$ is affine-invariant we have that $T_\alpha \in \PAut(\C^*)\subseteq\PAut(\C)$. Moreover, since $S_{G^*}\simeq S_{\Z_{r_1}\times\Z_{r_2}}$ (the group of automorphisms of $\Z_{r_1}\times\Z_{r_2}$) we have an isomorphism (induced by $\varphi$)
$$\phi:\langle T_\alpha \rangle\rightarrow \langle T_1,T_2\rangle $$
where $T_1(x,y)=(\overline{x+1},y)$ for any $(x,y)\in\Z_{r_1}\times\Z_{r_2}$, and $T_2$ is defined analogously.

We are proving that $\langle T_\alpha\rangle$ is an $s$-PD-like set for $\C^*$ with respect to $I$ for certain natural number $s$. First we need the following technical lemma proved in \cite{BS2}. We denote by $[\cdot]_r$ the remainder modulo $r$.

\begin{lemma}[\cite{BS2}]\label{juntarpuntos}
 Let $r, x_1,\dots,x_h\in \N$ where $0\leq x_1<x_2<\dots<x_h<r$. Then, there exists $\mu\in\N$ such that 
$$\left\lceil \frac{r}{h}\right\rceil-1\leq [x_i+\mu]_r<r$$
for all $i=1,\dots,h$ and $[x_j+\mu]_r=r-1$ for some $j\in \{1,\dots,h\}$.
\end{lemma}

Next result gives us the desired PD-like set. Moreover, in Section \ref{probability} we will discuss its proof in order to study the possibility of taking a smaller PD-like set.

\begin{theorem}[\cite{BS3}]\label{sPDlike}
The group generated by $T_\alpha$ is an $s$-PD-like set for $\C^*$ with respect to $I'$ (defined from $I$ as in (\ref{Iprima})) where
$$s=(\lambda_0+1) r_2-1$$
and $\displaystyle \lambda_0=\max\left\{\lambda\mid m<\left\lceil \frac{r_1}{\lambda}\right\rceil\right\}$.
\end{theorem}

\begin{proof}
  Let us consider $B\subseteq \{\alpha^i\}_{i=0}^{n-1}$, with $\left|B \right|=s$, and define $B'=\{\varphi(i):\alpha^i \in B\}\subseteq \Z_{r_1}\times\Z_{r_2}$, where $\varphi$ is the isomorphism fixed previously.  
	
	Then, we claim that there exists $\tau\in \langle T_1, T_2\rangle$ such that $\tau(B')\cap \Gamma=\emptyset$.
	
	Let us consider $\pi_i(B')$ the projection onto $\Z_{r_i}$ for $i=1,2$. It is clear that $B'=\bigcup\limits_{j\in \pi_2(B')} B_j$ (a disjoint union), where $$B_j=\{b\in B'\mid \pi_2(b)=j\},$$ for each $j\in \pi_2(B')$. 
	
	First, suppose that $\pi_2(B')\subsetneq \Z_{r_2}$. Then we have that there exists $\delta\in \N$ such that $\pi_2\left(T_2^\delta(B')\right)\subseteq \{1,\dots,r_2-1\}$ so $T_2^\delta(B')\cap \Gamma=\emptyset$ and we are done.
	
	Now, suppose that $\pi_2(B')= \Z_{r_2}$. If $\left|B_{j}\right|> \lambda_0$ for all $j\in\pi_2(B')$ then 
	 $$(\lambda_0+1) r_2-1=s=\left|B'\right|=\sum\limits_{j\in\pi_2(B')}\left|B_j\right|\geq (\lambda_0+1)\cdot r_2,$$
a contradiction. So, we conclude that there exists $j_0\in \pi_2(B')$ such that $\left|B_{j_0}\right|\leq \lambda_0$.  

  Then, let $\delta\in \N$ be such that $\pi_2\left(T_2^\delta(B_{j_0})\right)=0$ (and $\pi_2\left(T_2^\delta(B_{j})\right)\neq 0$ for any $j\neq j_0, j\in \Z_{r_2}$).
	
	By Lemma \ref{juntarpuntos}, applied to $r=r_1$ and the elements of $\pi_1\left(B_{j_0}\right)$, we have that there exists $\mu$ such that $\pi_1\left(T_1^\mu(B_{j_0})\right)\subseteq \{m,\dots, r_1-1\}$. So we conclude that
	$$T_1^\mu T_2^\delta(B')\cap\Gamma=\emptyset$$
	
	which finishes the proof.
\end{proof}

Finally, next theorem yields what we wanted.

\begin{theorem}\label{MainResult}
 Let $R_q(1,m)$ be the first-order Generalized Reed-Muller code of length $q^m$, ($m\in \N, q$ a power of a prime number). If $n=q^m-1=r_1\cdot r_2$, with $gcd(r_1,r_2)=1, r_1,r_2>1$, then we can correct up to $s$ errors by using Algorithm I with respect to the information set given in (\ref{infosetfirstorder}), where
 $$s=(\lambda_0+1)\cdot r_2-1$$
and $\lambda_0=max\{\lambda\mid m<\left\lceil \frac{r_1}{\lambda}\right\rceil\}$.
\end{theorem}

\begin{proof}
 By Theorem \ref{sPDlike} we have that the group generated by $T_\alpha$ is an $s$-PD-like set when we take as information set that given in (\ref{infosetfirstorder}), so we can apply the Algorithm I successfully.
\end{proof}

\section{Examples}\label{Examples}

In this section we are applying Theorem \ref{MainResult} to obtain some examples of the achieved error correction capability for some values of $q$ and $m$. Then we compare them with the references in the literature about the codes $R_q(1,m)$; they are the following:

\begin{itemize}
	\item In \cite{KMM2} the authors mention that the group of translations in $\K^m$ is an $s$-PD-set for $R_q(v,m)$ where 
	$$s=\min\left\{\left\lfloor\frac{q^m-1}{f_{v,m,q}} \right\rfloor, \left\lfloor\frac{d_{v,m,q}-1}{2}\right\rfloor\right\}$$
	and $f_{v,m,q}, d_{v,m,q}$ represent the dimension and the minimum distance  of $R_q(v,m)$ respectively.
	
	In the context of this paper, we are interested only in the case of first-order GRM codes. It is well-known that for $R_q(1,m)$ one has $f_{v,m,q}=m+1$ and $d_{v,m,q}=q^{m-1}\cdot (q-1)$. Then, we denote:
	$$\rho_1=\frac{q^m-1}{m+1},\qquad \rho_2=\frac{q^{m-1}\cdot(q-1)-1}{2}$$
	
	\item In \cite{KMM3} the authors explain that for the codes $R_q(1,m)$, $s$-PD-sets of size $s+1$ exist for 
	$$1\leq s\leq \left\lfloor\frac{q^m}{m+1}\right\rfloor-1.$$
	Let us denote $\displaystyle \rho_3=\frac{q^m}{m+1}$.
\end{itemize}

Next proposition shows that we have to pay attention to the value $\rho_1$ uniquely.

\begin{proposition}
 For any $q, m\geq 3$ the following hold:
\begin{enumerate}
	\item $\displaystyle \left\lfloor\rho_3 \right\rfloor=\left\lfloor\rho_1 \right\rfloor$ or $\displaystyle \left\lfloor\rho_1 \right\rfloor+1$.
	\bigskip
	\item $\displaystyle \rho_2 \geq \rho_1 $.
\end{enumerate}
\end{proposition}
\begin{proof}
\begin{enumerate}
	\item In case $\displaystyle \rho_3=\frac{q^m}{m+1}\in\Z$ we have 
	$$\left\lfloor\frac{q^m-1}{m+1} \right\rfloor=\left\lfloor\frac{q^m}{m+1}-\frac{1}{m+1} \right\rfloor=\frac{q^m}{m+1}-1=\left\lfloor\frac{q^m}{m+1}\right\rfloor-1=\left\lfloor\rho_3\right\rfloor-1.$$
	In case $\displaystyle \rho_3=\frac{q^m}{m+1}\notin\Z$ we have
	$$\left\lfloor\frac{q^m-1}{m+1} \right\rfloor=\left\lfloor\frac{q^m}{m+1}-\frac{1}{m+1} \right\rfloor=\left\lfloor\frac{q^m}{m+1}\right\rfloor=\left\lfloor\rho_3\right\rfloor.$$
	\item Using the notation introduced above we have 
	$$\rho_2=\frac{q^m-1}{2}-\frac{q^{m-1}}{2}=\rho_1\cdot\frac{m+1}{2}-\frac{q^{m-1}}{2}.$$
	So, $\rho_2\geq\rho_1$ if and only if 
	$$\rho_1\frac{m+1}{2}-\frac{q^{m-1}}{2}\geq\rho_1\Leftrightarrow \rho_1\left(\frac{m+1}{2}-1\right)\geq \frac{q^{m-1}}{2}\Leftrightarrow$$
	$$\rho_1\frac{m-1}{2}\geq\frac{q^{m-1}}{2}\Leftrightarrow \rho_1\geq \frac{q^{m-1}}{m-1}$$
	which is equivalent to
	$$\frac{q^m-1}{q^{m-1}}\geq\frac{m+1}{m-1}.$$
	
	Now, let us consider the real functions $\displaystyle f(x)=\frac{q^x-1}{q^{x-1}},\ g(x)=\frac{x+1}{x-1}$ defined in its corresponding domains. On the one hand, we have that the first derivative function of $f(x)$ is
	$$f'(x)=\frac{q^{x-1}\cdot\ln q}{q^{2(x-1)}},\text{ for any real number } x$$
	It is clear that $f'(x)\geq 0$ for any real number $x$ (and any value of $q$, a power of a prime number) so the function $f(x)$ is an increasing one.
	
	On the other hand,
	$$g'(x)=\frac{-2}{(x-1)^2}<0, \text{ for any real number } x\neq 1$$
	so $g(x)$ is a decreasing function. 
	
	Finally, we want to see that $f(3)\geq g(3)$ for any value of $q\geq 3$. We have that $\displaystyle f(3)=\frac{q^3-1}{q^2}$ and $g(3)=2$, so we consider the real function $\displaystyle h(x)=\frac{x^3-1}{x^2}$ defined for any real number $x\neq 0$. First, $h(3)=26/9>2$, and second, it is easy to see that $h(x)$ is an increasing function in the interval $[3,+\infty)$; so we have what we wanted.
	
	This finishes the proof.
	
	\end{enumerate}
\end{proof}

\begin{corollary}
For first-order GRM codes $R_q(1,m), (q>2)$ the best bound for the error correction capability between $\{\rho_1,\rho_2,\rho_3\}$ will always be 
$$\rho_1=\left\lfloor\frac{q^m-1}{m+1}\right\rfloor.$$
\end{corollary}

Therefore, we have to compare the values of $s$ obtained from Theorem \ref{MainResult} with $\rho_1$ uniquely. 

In Table \ref{TablaRes3}, Table \ref{TablaRes4} and Table \ref{TablaRes5} we include examples, for some values of $m$, for the cases $q=3, q=4$ and $q=5$ respectively. Although for certain values of $m$ (most of them) we can find several decompositions $n=r_1\cdot r_2$ we only show that yield the best error-correction capability. In column $s$ we include the value obtained from Theorem \ref{MainResult}.

\begin{table}[h]
\caption{Number of corrected errors for $R_3(1,m)$}
\label{TablaRes3}
\centering

\begin{tabular}{|c|c|c|c|c|}
 \hline &&&&\\
 $m$ & $r_1$ & $r_2$ &$\rho_1$& $s$\\
&&&&\\ \hline
3&13&2&6&9\\ \hline
4&5&16&16&31\\ \hline
5&121&2&40&49\\ \hline
6&7&104&207&104\\ \hline
7&1093&2&273&313\\ \hline
8&41&160&728&959\\ \hline
9&757&26&1968&2209\\ \hline
10&61&968&5368&6775\\ \hline
\end{tabular}
\end{table}

\begin{table}[h]
\caption{Number of corrected errors for $R_4(1,m)$}
\label{TablaRes4}
\centering

\begin{tabular}{|c|c|c|c|c|}
 \hline &&&&\\
 $m$ & $r_1$ & $r_2$ &$\rho_1$& $s$\\
&&&&\\ \hline
2&5&3&5&8\\ \hline
3&7&9&15&26\\ \hline
4&17&15&54&74\\ \hline
5&11&93&170&278\\ \hline
6&13&315&585&944\\ \hline
7&43&381&2047&2666\\ \hline
8&257&255&7281&8414\\ \hline
9&19&13797&26214&41390\\ \hline
10&41&25575&95325&127874\\ \hline
\end{tabular}
\end{table}

\begin{table}[h]
\caption{Number of corrected errors for $R_5(1,m)$}
\label{TablaRes5}
\centering

\begin{tabular}{|c|c|c|c|c|}
 \hline &&&&\\
 $m$ & $r_1$ & $r_2$ &$\rho_1$& $s$\\
&&&&\\ \hline
2&3&8&8&15\\ \hline
3&31&4&31&43\\ \hline
4&13&48&124&191\\ \hline
5&11&284&520&851\\ \hline
6&7&2232&2232&4463\\ \hline
7&19531&4&9765&11163\\ \hline
8&313&1284&43402&49919\\ \hline
9&19&102796&195312&308387\\ \hline
10&33&295928&887784&1183711\\ \hline
\end{tabular}
\end{table}

\section{A probabilistic approach to the problem of finding smaller PD-like sets}\label{probability}

The arguments in the proof of Theorem \ref{sPDlike} show that in many cases we will be able to correct by using a notably smaller s-PDlike set. In this section we study, from a probabilistic point of view, when this occurs.

Let us denote $r=c+e$ the received word, where $c\in\C$ and $e$  %$e=bX^0+\sum\limits_{i=0}^{n-1}e_i X^{\alpha^i}\in\K G$ 
 the error vector. We write $\supp(e)$ to denote the positions that correspond to the non-zero coefficients of $e$ in its polynomial expression (see (\ref{polynomial})). If $0\in\supp(e)$ then we can not correct by using the PD-like set given in Section \ref{NewPD}. However, the structure of Algorithm I assures that by applying sequently the elements of $\Sigma$ (see Step 3), we achieve a situation where $0$ does not belong to the suppport of the modified error. So, let us assume that $0\notin\supp(e)$.

Recall that $n=r_1\cdot r_2$ and we have fixed $\alpha\in G^*$, an $n$-th primitive root of unity, and $\varphi$, an isomorphism $\varphi: \Z_n\longrightarrow \Z_{r_1}\times\Z_{r_2}$. For any $\alpha^i\in B$, we write $\pi_j(\varphi(i))$ to denote the projection of $\varphi(i)$ onto $\Z_{r_j}$, for $j=1,2$. We want to correct up to $s$-errors, where $s=(\lambda_0+1)\cdot r_2-1$ as it was given in Theorem \ref{sPDlike}.

If we analyze the proof of the mentioned result we can see that in case 
\begin{equation}\label{condicionT2}\{\pi_2(\varphi(i))\}_{\alpha^i\in B}\subsetneq\Z_{r_2}\end{equation}
we are able to move the errors out of the information set $I$ by using the group $\langle T_2\rangle$ uniquely instead of the entire group $\langle T_1,T_2,\rangle$. So, as a first consequence, we have that in order to apply Algorithm I as much efficiently as possible we take the elements in the group $\langle T_1,T_2\rangle$ by starting with the elements of $\langle T_2\rangle$.

Now, given a random error vector, we deal with the problem of determining the probability that it verifies condition (\ref{condicionT2}). To compute it, let us denote $A=\{S\subseteq \Z_{r_1}\times\Z_{r_2}\mid\ |S|=s, \text{ and } \pi_2(S)\subsetneq \Z_{r_2}\}$, and for any $i\in\Z_{r_2}=\{0,1,\dots,r_2-1\}$, we write $A_i=\{S\subseteq \Z_{r_1}\times\Z_{r_2}\mid\ |S|=s, \text{ and } i\notin\pi_2(S)\}$. Then, it is clear that $A=\bigcup\limits_{i=0}^{r_2-1} A_i$.

Next theorem shows a formula to calculate the desired probability. Recall that $\phi:\langle T_\alpha \rangle\rightarrow \langle T_1,T_2\rangle $ is the isomorphism induced by $\varphi$.

\begin{theorem}
In order to apply Algorithm I, the group $\phi^{-1}\left(\langle T_2\rangle\right)$ acts as an $s$-PD-like set with probability 
$${\mathbf p}=\frac{\left|A\right|}{{n\choose s}}=\frac{1}{{n\choose s}}\cdot \sum\limits_{j=1}^{r_2-\delta}(-1)^{j+1}\cdot\left(\sum\limits_{\{i_1,\dots,i_j\}\subseteq\Z_{r_2}}\left|A_{i_1}\cap\cdots\cap A_{i_j}\right|\right)$$
where $\displaystyle\delta=\left\lceil \frac{s}{r_1}\right\rceil$.
\end{theorem}

\begin{proof}
 Let  $B\subseteq \{\alpha^i\}_{i=0}^{n-1}$, with $\left|B \right|=s$, and define $B'=\{\varphi(i):\alpha^i \in B\}\subseteq \Z_{r_1}\times\Z_{r_2}$. The proof of Theorem \ref{sPDlike} shows that the group generated by $T_2$ acts as an $s$-PD-like set when $\pi_2(B')\subsetneq\Z_{r_2}$. 

So, given a random subset in $\Z_{r_1}\times\Z_{r_2}$, with $s$ elements, we want to calculate the probability $\mathbf p$ of it belonging to $A$. It is clear that the number of subsets of cardinality $s$ in $\Z_{r_1}\times\Z_{r_2}$ is ${n\choose s}$ and then
${\mathbf p}=\frac{\left|A\right|}{{n\choose s}}$.

Now, from the equality $A=\bigcup\limits_{i=0}^{r_2-1} A_i$ we have that

$$|A|=\sum\limits_{j=1}^{r_2-1}(-1)^{j+1}\cdot\left(\sum\limits_{\{i_1,\dots,i_j\}\subseteq Z_{r_2}}|A_{i_1}\cap\cdots\cap A_{i_j}|\right)$$

On the other hand, we have that, since any subset contained in $A$ has cardinality $s$, some intersections in the formula above has cardinality zero. Specifically, if $\displaystyle\delta=\left\lceil \frac{s}{r_1}\right\rceil$ then 
$$|A_{i_1}\cap\cdots\cap A_{i_j}|=0$$
for any $j\in\{r_2-\delta+1,\dots,r_2-1\}$. Observe that given $i\in\Z_{r_2}$ we have that $|\{b\in B\mid \pi_2(\varphi(b))=i\}|\leq r_1$ This finishes the proof.

\end{proof}

We have applied the previous result in order to get some explicit values of the probability $\mathbf p$. In Table \ref{TablaProb} we include the most relevant ones for the corresponding values of $q, m, r_1, r_2$ and $s$.

\begin{table}[h]
%\begin{center}
\caption{Probability for some values of $q$ and $m$}\label{TablaProb}
%\end{center}

%\begin{center}
\begin{tabular}{|c|c|c|c|c|c|}
\hline
$q$&$m$&$r_1$&$r_2$&$s$&$\mathbf p$\\ \hline
$3$&$6$&$7$&$104$&$207$&$\simeq 0.9999$\\ \hline
$\cdot\cdot$&$8$&$32$&$205$&$819$&$\simeq 0.9482$\\ \hline
$5$&$5$&$11$&$284$&$851$&$\simeq 0.9999$\\ \hline
$\cdot\cdot$&$6$&$7$&$2232$&$4463$&$\simeq 1$\\ \hline
$4$&$5$&$11$&$93$&$278$&$\simeq 0.9516$\\ \hline
$\cdot\cdot$&$6$&$13$&$315$&$944$&$\simeq 0.9999$\\ \hline

\end{tabular}
%\end{center}
\end{table}

As this table shows we can expect that in so many cases we will able to use as $s$-PD-like set the group $\phi^{-1}\left(\langle T_2\rangle\right)$ instead of the entire group of translations.

\end{document}